\documentclass[10pt]{article}

\usepackage{amssymb,amsmath}
\usepackage{mathrsfs,color}
\usepackage[all]{xypic}
\usepackage{graphicx}
\usepackage{tikz,color} 
\usetikzlibrary{calc,matrix,arrows,chains,positioning,scopes,decorations.pathmorphing,backgrounds,fit}

\usepackage{yfonts}

\usepackage{enumitem}
\usepackage{adjustbox}
\usepackage{thmtools,mathtools}
\usepackage{ulem}

\usepackage{nameref,
cleveref}

\usepackage[utf8]{inputenc}
\usepackage{stackrel}
\usepackage{dcpic}
\usepackage{pictexwd}
\usepackage{fixltx2e}
\usepackage{enumerate}
\usepackage{relsize}
\usepackage{stmaryrd}
\usepackage{cancel}
\usepackage{yfonts}
\usepackage{xcolor}
\usepackage{csquotes }
 \usepackage{amsthm}

\makeatletter
\newcommand\xLeftrightarrowceva[2][]{%
  \ext@arrow 9999{\longLeftrightarrowfill@}{#1}{#2}}
\newcommand\longLeftrightarrowfill@{%
  \arrowfill@\Leftarrow =\Rightarrow}
  \makeatother
\makeatletter

\DeclareRobustCommand\mos[1]{\mathrel{|}\joinrel
\stackrel{#1}{\mathrel{=}}}

\DeclareRobustCommand\vds[1]{\,\mathrel{|}\joinrel\joinrel\joinrel
\frac{#1}{ \ \ \ }}

\makeatletter

\newcommand\vso[2][]{%
  \ext@arrow 9999{\vsfill@}{#1}{#2}}
\newcommand\vsfill@{%
  \arrowfill@ \mid\ -}
\makeatother

\newcommand{\setmod}[1]{\textsc{#1}}
\newcommand{\logl}{\mathbf{K}\Lambda}

\newcommand{\ri}{\rightarrow}
\newcommand{\tr}{\sigma}

\newcommand{\foss}{Form_s}
\newcommand{\srb}{\sigma}
\newcommand{\vp}{\varphi}
\newcommand{\si}[2]{\Sigma_{{#1}_1 \cdots {#1}_{#2}, {#1} }}

\newcommand{\mb}{\scriptscriptstyle{\Box}}
\newcommand{\spt}{\sigma^{\mb}}

\theoremstyle{plain}
\newtheorem{theorem}{Theorem}[section]
\newtheorem{lemma}[theorem]{Lemma}

\newtheorem{remark}[theorem]{Remark}
\newtheorem{proposition}[theorem]{Proposition}
\newtheorem{definition}[theorem]{Definition}
\newtheorem{corollary}[theorem]{Corollary}

\newtheorem{example}[theorem]{Example}

\newcommand{\RNum}[1]{\uppercase\expandafter{\romannumeral #1\relax}}

\newcommand{\lra}{\leftrightarrow}

\begin{document}

\normalem

  \title{A many-sorted polyadic modal logic}
  \author{Ioana Leu\c stean, Natalia Moang\u a  and Traian Florin \c Serb\u anu\c t\u a\\
{\small Faculty of Mathematics and Computer Science, }\\
{\small University of Bucharest, Bucharest, Romania} \\
{\small ioana@fmi.unibuc.ro}  \hspace*{0.3cm}  {\small natalia.moanga@drd.unibuc.ro  \hspace*{0.3cm} traian.serbanuta@fmi.unibuc.ro}  
}
\date{}

\maketitle
 
\begin{abstract}

  {We propose a general system that combines the powerful features of modal logic  and many-sorted reasoning. Its algebraic semantics leads to a many-sorted generalization of boolean algebras with operators, for which we prove the analogue of the J\'{o}nsson-Tarski theorem.  Our goal was to deepen the connections between modal logic and program verification, and we test the expressivity of our system by defining a small imperative language and its operational semantics.}

\medskip 

\noindent {\em Keywords}: Polyadic modal logic, many-sorted logic,  boolean algebras with operators, many-sorted algebras,   J\'{o}nsson-Tarski theorem, operational semantics.
\end{abstract}

\section{Introduction}

In this paper we define a many-sorted polyadic modal logic, together with its corresponding  algebraic theory. 
The idea is not new: in \cite{manys3,manys2} two-sorted systems are analyzed and we used them as references for our approach, while in
\cite{manys1,manys6, manys4} a general theory is developed in a coalgebraic setting.  However, to our knowledge, the general framework presented in this paper is new.   

 Our language is determined by a fixed, but arbitrary,  many-sorted signature and a set of many-sorted propositional variables. The transition from mono-sorted to many-sorted setting is a smooth process and we follow closely the developments from \cite{mod}.
We define appropriate frames and models, the generalized construction of the canonical model and we prove the completeness of our results. The distinction between local and global deduction from the mono-sorted setting is deepened in our version: locally we consider only hypotheses of the same sort, while globally the set of hypotheses is a many-sorted set. The global deduction is analyzed in a distinct section, where we also prove a corresponding generalization of the deduction theorem. 
In order to investigate the algebraic completeness, we introduce a many-sorted generalization of boolean algebras with operators and we prove the analogue of J\'{o}nsson-Tarski theorem. We mention that similar structures were defined in \cite{manys1}, but in that case the operators are unary operations while, in our setting, they have arbitrary arities.

While the transition from the mono-sorted logic to a many-sorted one  is a smooth process, we see our system as a step towards deepening the connection between modal logic and program verification, and we test its expressivity in the last section. 

Our research was inspired by \textit{Matching logic} \cite{rosu},  a logic for program specification and verification, in which one can represent the mono-sorted polyadic modal logic. The propositional calculus defined in this paper can be seen as the propositional counterpart of Matching logic.

In Section \ref{mainlog} we follow a standard approach in order to define the canonical model and to prove the completeness theorem. In this section only the local deduction is considered. Section \ref{secglobal} is dedicated to global deduction. Section \ref{secalg} contains the algebraic semantics. In Section \ref{exlog}  we relate our system  with Matching logic and we further compare our approach  with similar systems  from \cite{manys2,manys3,contextModal}. Finally, in the last section, we define within our logic, both the syntax and the operational semantics of a small imperative language such that program execution is modeled as logical
inference.

\section{The many-sorted polyadic modal logic ${\mathcal M}{\mathcal L}_S$}\label{mainlog}

In this section we follow closely the development of the polyadic modal logic from \cite{mod}. Fixing a many-sorted signature $(S,\Sigma)$ we investigate the general theory of many-sorted modal logics based on $(S, \Sigma)$. For these systems we define the syntax, the semantics, the local deduction, the canonical model and we prove the completeness results. Our system is a generalization of the two-sorted modal logic defined in \cite{manys2}.   {Some proofs in the paper are straightforward generalizations of the mono-sorted case; however we sketch them in order to  keep the paper self-contained.}  
\subsection{Formulas, frames and models}

Let $(S, \Sigma)$ be a many-sorted signature. A set of {\em $S$-sorted variables} is an $S$-sorted set $P=\{P_s\}_{s\in S}$
 such that  $P_s\neq \emptyset$ for any $s\in S$ and 
 $P_{s_1} \cap P_{s_2} = \emptyset $ for any  $s_1 \neq s_2$ in $S$.

An $(S,\Sigma)$-{\em modal language} $\mathcal{ML}_{(S,\Sigma)}(P) $ is built up using the many-sorted signature $(S,\Sigma)$ and a set of propositional variables $P$. 
 
 In the sequel we assume $(S, \Sigma)$ and $P$ are fixed. For brevity, $\mathcal{ML}_{(S,\Sigma)}(P) $ will be denoted $\mathcal{ML}_S$. For any $n\in {\mathbb N}$ we denote  $[n]:=\{1,\ldots, n\}$ and $\Sigma_{s_1\ldots s_n,s}=\{\sigma\in \Sigma\mid \sigma:s_1
\ldots s_n\to s\}$ for any $s,s_1,\ldots, s_n\in S$.

\begin{definition}
The set of formulas of $\mathcal{ML}_S$ is an $S$-indexed family $Form_S=\{Form_s\}_{s\in S}$  inductively defined as follows:
 
\begin{itemize}
\item $P _s \subseteq  \foss$ for any $s\in S$,
\item if $\phi_1 \in \foss$ then $\neg \phi_1 \in \foss$ for any $s\in S$,
\item if $\phi_1, \phi_2 \in \foss$  then $\phi_1 \vee \phi_2 \in \foss $ for any $s\in S$,
\item if $\srb \in \Sigma_{s_1 \ldots s_n,s}$ and $\phi_1 \in Form_{s_1},\ldots, \phi_n\in Form_{s_n}$  then $\srb (\phi_1, \ldots , \phi_n ) \in \foss$.
\end{itemize}
\end{definition}

As usual, for any $s\in S$, $\phi_1$, $\phi_2\in Form_s$  we set
$\phi_1\wedge\phi_2 :=\neg(\neg\phi_1\vee\neg\phi_2)$ and $\phi_1\ri\phi_2 :=\neg\phi_1\vee\phi_2$. For any $s\in S$ and a fixed $p\in P_s$ we define
$\perp_s:= p\wedge \neg p$ and $\top_s=\neg\perp_s$. 
Moreover,  if $\srb \in \Sigma_{s_1 \ldots s_n,s}$  is a non-nullary operation and $\phi_i\in Form_{s_i}$ for any $i \in  [n] $, the {\em dual operation} is   
\begin{center}
$\spt (\phi_1, \ldots , \phi_n ) := \neg\srb (\neg\phi_1, \ldots , \neg\phi_n ).$
\end{center}

In order to define the semantics we introduce the  $(S,\Sigma)${\em -frames} as $S$-sorted relational structures.

\begin{definition}
 An $(S,\Sigma)${\em -frame} is a tuple 
 $\mathcal{F} =(\mathcal{W}, \mathcal{R})$
 such that:
\begin{itemize}
\item  $\mathcal{W} =\{ W_s \}_{s\in S}$ is an  $S$-sorted set of worlds and $W_s\neq \emptyset$ for any $s\in S$,
\item $\mathcal{R}=\{\mathcal{R}_{\srb}\}_{\srb \in \Sigma}$ such that $\mathcal{R}_{\srb} \subseteq  W_s \times W_{s_1} \times \ldots \times W_{s_n}$  for any $\srb \in \Sigma_{s_1 \ldots s_n,s}$.
\end{itemize}

If $\mathcal{F}$ is an $(S,\Sigma)$-frame, then an $(S,\Sigma)${\em -model based on} $\mathcal{F}$ is a pair     \mbox{$\mathcal{M}= (\mathcal{F}, \rho)$} where \mbox{ $\rho: P \ri \mathcal{P(W)} $} is an {\em $S$-sorted valuation function} such that \mbox{$\rho_s : P_s \ri \mathcal{P}(W_s)$} for any $s\in S$. The model $\mathcal{M}= (\mathcal{F}, \rho)$ will be simply denoted by  $\mathcal{M}= (\mathcal{W}, \mathcal{R}, \rho)$. 

Following \cite{mod},  if $\setmod{C}$ is a set of frames then we say that a model $\mathcal{M}$ {\em  is from} $\setmod{C}$ if 
it is based on a frame from $\setmod{C}$.
\end{definition}

In the sequel we define the {\em  satisfaction relation}.

\begin{definition}
Let $\mathcal{M}= (\mathcal{W}, \mathcal{R}, \rho)$ be an $(S,\Sigma)$-model, $s\in S$, $w\in W_s$ and  $\phi \in \foss$. We define 
\vspace*{-0.2cm}
$$\mathcal{M},w\mos{s} \phi$$
by induction over $\phi$ as follows:

\begin{itemize}
\item $\mathcal{M},w \mos{s} p$ iff $w\in \rho_s(p)$
\item $\mathcal{M},w \mos{s} \neg \psi$ iff $\mathcal{M},w \not\mos{s}\psi$
\item $\mathcal{M},w \mos{s} \psi_1 \vee \psi_2$ iff $\mathcal{M},w \mos{s} \psi_1$ or $\mathcal{M},w \mos{s} \psi_2$ 
\item if $\srb \in \Sigma_{s_1 \ldots s_n,s} $  then $\mathcal{M},w \mos{s} \srb (\phi_1, \ldots , \phi_n )$ iff there exists  $(w_1,\ldots,w_n) \in W_{s_1}\times\cdots\times W_{s_n}$  such that $\mathcal{R}_{\srb} ww_1\ldots w_n$ and
 $\mathcal{M},w_i  \mos{s_i} \phi_i$ for any $i \in  [n] $.
\end{itemize}
\end{definition}

We note that:

\begin{itemize}
\item[(1)] When $\srb \in \Sigma_{\lambda ,s}$  we have  $\mathcal{M},w \mos{s} \tr$ iff $w \in \mathcal{R}_{\srb}$.
\item[(2)] For $\srb \in \Sigma_{s_1 \ldots s_n,s}$ we have
$\mathcal{M},w \mos{s} \spt (\phi_1, \ldots , \phi_n )$ iff $\mathcal{R}_{\srb} ww_1\ldots w_n$  implies  $\mathcal{M},w_i \mos{s_i} \phi_i$ for some $i \in  [n] $
for any $(w_1,\ldots,w_n) \in W_{s_1}\times\cdots\times W_{s_n}$ 
 \end{itemize}

In this section we study the local deduction, while the global one will be investigated in Section \ref{secglobal}. Note that in our sorted setting the local approach implies  that the set of hypothesis and the conclusion have the same sort. Therefore we explicitly write the sort in the definition of the local deduction.

\begin{definition}
 Let $s\in S$, $\phi \in\foss$ and assume ${\mathcal M}$ is an $(S,\Sigma)$-model. We say that $\phi$ is  {\em  universally true} in  $\mathcal{M}$, and we denote $\mathcal{M} \mos{s} \phi$, if  $\mathcal{M},w \mos{s} \phi$ for all $w\in W_s$. If $\Phi_s\subseteq Form_s$, $\mathcal{M} \mos{s} \Phi_s$ is defined as usual. We say that $\Phi_s$ is {\em satisfiable} if there is a model $\mathcal{M}$ and $w\in W_s$ such that $\mathcal{M},w\mos{s} \Phi_s$.

Assume that $\setmod{C}$ is a class of frames or a class of models. If $s\in S$ and  $\Phi_s\cup\{\phi\}\subseteq Form_s$ then we say that $\phi$ is a {\em local semantic consequence of} $\Phi_s$ {\em  over} $\setmod{C}$, and we write $\Phi_s\mos{s}_\setmod{C}\phi$ if $\mathcal{M},w\mos{s} \Phi_s$ implies $\mathcal{M},w\mos{s} \phi$ for any $(S,\Sigma)$-model $\mathcal M$ from $\setmod{C}$  and for any $w\in W_s$ (we simply denote $\Phi_s\mos{s}\phi$ when $\setmod{C}$ is the class of all frames). 
\end{definition}


\subsection{The deductive system}\label{mainlogded}
Recall that in our setting any variable uniquely  determines its sort, and the sort of a formula is uniquely determined by the sorts of its variables.
Consequently, the {\em  uniform substitution} is $S$-sorted, meaning that a variable of sort $s\in S$ will be uniformly replaced only by a formula of the same sort.
Moreover, for any theorem of classical logic, there exists a corresponding one in $Form_s$ for any sort $s\in S$.

 We now define ${\mathbf K}_{(S,\Sigma)}$, a generalization of the modal system $\mathbf K$ (see \cite{mod} the the mono-sorted version). Hence ${\mathbf K}_{(S,\Sigma)}=\{{\mathbf K}_s\}_{s\in S}$ is the least $S$-sorted set of formulas with the following properties:
\begin{itemize}
\item[(a0)] for any $s\in S$, if $\alpha\in Form_s$ is a theorem in classical logic, then $\alpha\in {\mathbf K}_s$, 
\item[(a1)]  the following formulas are in ${\mathbf K}_s$

\medskip

\begin{itemize}
\item[($K^i_\sigma$)] $\,\,\,\spt(\psi_1, \ldots , \phi\rightarrow \chi, \ldots , \psi_n) \ri $

\hspace*{2cm}$\ri ( \spt (\psi_1, \ldots , \phi, \ldots , \psi_n) \ri \spt (\psi_1, \ldots ,\chi , \ldots , \psi_n))$

\medskip
\item[($Dual_\sigma$)] $\srb (\psi_1,\ldots ,\psi_n )\leftrightarrow \neg \spt (\neg \psi_1,\ldots ,\neg \psi_n )$ for any $n\geq 1$, $i\in [n]$,\\  \mbox{  $\srb\in \Sigma_{s_1\cdots s_n,s}$,$\psi_{s_1} \in Form_{s_1} $, $\ldots$, $\psi_{s_n} \in Form_{s_n}$}
and \mbox{$\phi,\chi\in Form_{s_i}$.}

\end{itemize}
\end{itemize}

The  deduction rules are  {\em  Modus Ponens} and {\em   Universal Generalization}:
\medskip

{\bf(MP$_s$)}$\,\,\,\displaystyle\frac{\phi\,\,\,\phi\to\psi}{\psi}$
where $s\in S$ and $\phi,\psi\in Form_s$

{\bf (UG$_\sigma^i$)} 
$\,\,\,\displaystyle\frac{\phi}{\spt (\psi_{s_1}, \ldots \psi_{i-1},\phi, \psi_{i+1},\ldots, \psi_{s_n})}$ where $n\geq 1$,   $\srb\in \Sigma_{s_1\cdots s_n,s}$, 

\medskip

\mbox{$i \in  [n] $},  $\srb\in \Sigma_{s_1\cdots s_n,s}$, $\psi_{s_1} \in Form_{s_1} $, $\ldots$, $\psi_{s_n} \in Form_{s_n} $ 
and $\phi\in Form_{s_i}$

\medskip

(the n-place operator $\spt$ is associated with $n$ generalization rules).
\medskip

Since $(S,\Sigma)$ is fixed, we simply write $\mathbf K$  instead of ${\mathbf K}_{(S,\Sigma)}$.

\begin{definition}\label{clsub}
Let $\Lambda\subseteq Form_S$ be an $S$-sorted set of formulas. The  {\em normal modal logic} defined by $\Lambda$ is ${\mathbf K}\Lambda=\{{\mathbf K}\Lambda_s\}_{s\in S}$ where
\begin{center}
${\mathbf K}\Lambda_s:={\mathbf K}_s\cup\left\{\lambda^\prime\in Form_s\mid  \lambda^\prime \mbox{ is obtained by uniform substitution}\right.$

 $\left. \mbox{\hspace*{6cm}applied to a formula }\lambda\in \Lambda_s \right\}$
\end{center}
\end{definition}

In the sequel we assume $\Lambda\subseteq Form_s$ is an $S$-sorted set of formulas and we investigate the normal modal logic $\mathbf{K}\Lambda$. Note that, in our approach, a logic is defined by its axioms.

\begin{definition}
Assume that $n\geq 1$, $s_1,\ldots, s_n\in S$ and  $\phi_i\in Form_{s_i}$ for any $i \in  [n] $. The sequence $\phi_1,\ldots, \phi_n$ is a {\em $\logl$-proof for} $\phi_n$  if,  for any $i \in  [n] $,  $\vp_i$ is in $\logl_{s_i}$  or $\vp_i$ is inferred from $\vp_1,\ldots, \vp_{i-1}$ using {\em modus ponens} and {\em universal generalization}. If $\phi$ has a proof in $\logl$ then we say that $\phi$ is a {\em theorem} and we write
$\vds{s}_{\logl}\phi$ where $s$ is the sort of $\phi$. 
\end{definition}

\begin{proposition}\label{genprop} 
Let $\sigma\in\Sigma_{s_1\ldots s_n,s}$, $i\in[n]$,  $\varphi, \phi \in Form_{s_i}$  and  $\psi_{s_j} \in Form_{s_j}$ for any $j\in  [n] $. The following hold:
\begin{itemize}
\item[(i)] \label{ded1}  ${\vds{s_i}}_{\logl} \phi\ri \varphi$
implies 

\hspace*{1.5cm}${\vds{s}}_{\logl} 
\spt(\ldots,\psi_{i-1} ,\phi,\psi_{i+1},\ldots)\ri
\spt (\ldots ,\psi_{i-1} ,\varphi, \psi_{i+1},\ldots)  $

\item[(ii)] \label{axsquare}
 ${\vds{s}}_{\logl}  \spt(\ldots,\psi_{i-1} , \phi\wedge \varphi,\psi_{i+1}, \ldots ) \leftrightarrow$
 
 \hspace*{1.5cm} $\leftrightarrow ( \spt (\ldots,\psi_{i-1} , \phi,\psi_{i+1}, \ldots ) \wedge \spt (\ldots,\psi_{i-1} ,\varphi ,\psi_{i+1}, \ldots )) $ 
\item[(iii)] 
${\vds{s}}_{\logl}  \srb(\ldots,\psi_{i-1} , \phi\vee \varphi,\psi_{i+1}, \ldots ) \leftrightarrow$
 
 \hspace*{1.5cm} $\leftrightarrow ( \srb (\ldots,\psi_{i-1} , \phi,\psi_{i+1}, \ldots ) \vee \srb (\ldots,\psi_{i-1} ,\varphi ,\psi_{i+1}, \ldots )) $
\item[(iv)] ${\vds{s_i}}_{\logl} \phi\leftrightarrow \varphi$
implies 

\hspace*{1.5cm}${\vds{s}}_{\logl} 
\srb(\ldots,\psi_{i-1} ,\phi,\psi_{i+1},\ldots)\leftrightarrow
\srb (\ldots ,\psi_{i-1} ,\varphi, \psi_{i+1},\ldots)  $

\end{itemize}
\end{proposition}

\begin{remark}\label{alt}
As in the mono-sorted case, the logic $\logl$ can be defined replacing ($K$) and ($Dual$) with 

$\srb(\psi_{1},  \ldots , \psi_{n}) \leftrightarrow \perp_{s}$ iff $\psi_i = \perp_{s_i}$ for some $i\in [n]$
 
$ \srb(\psi_{1}, \ldots , \phi\vee \varphi, \ldots , \psi_{n}) \leftrightarrow ( \srb (\psi_{1}, \ldots , \phi, \ldots , \psi_{n}) \vee \srb (\psi_{1}, \ldots ,\varphi , \ldots , \psi_{n})) $

\noindent  and ({\bf UG}) with 
 
  \mbox{$ \vds{s_i} _{\logl} \phi\ri \varphi$ implies\\$\vds{s}_{\logl} \srb (\psi_{1}, \ldots ,\phi, \ldots ,\psi_{n}) \ri \srb (\psi_{1}, \ldots ,\varphi, \ldots ,\psi_{n})  $}
  \medskip
   
 \noindent for any $n\geq 1$, $i\in [n]$,  $\srb\in \Sigma_{s_1\cdots s_n,s}$, $\psi_{s_1} \in Form_{s_1} $, $\ldots$, $\psi_{s_n} \in Form_{s_n}$ , $\phi,\chi\in Form_{s_i}$
\end{remark}

\begin{definition} \label{dedhyp} If $s\in S$, $\Phi_s \subseteq Form_s$  and $\phi\in Form_s$ then we say that $\phi$ is {\em locally provable from} $\Phi_s$ {\em in} $\logl$, and we write $\Phi_s{\vds{s}}_{\logl}\phi$,  if there are $\phi_1, \ldots, \phi_n\in \Phi_s$  such that   
$\vds{s}_{\logl}(\phi_1\wedge \ldots \wedge \phi_n) \ri \phi$. 
\end{definition}

Note that, apart from the usual features of local deduction, in our setting locality also implies all the hypothesis and the conclusion must have the same sort. We can state now the local version of the deduction theorem. The proof is an easy generalization of its analogue from classical logic.

\begin{theorem}(Local deduction theorem for ${\mathbf K}\Lambda$)\label{locdedth}

\begin{center} 
$\Phi_s\vds{s}_{\logl}\vp \ri \psi$  iff   $\Phi_s \cup \{\vp\}_s\vds{s}_{\logl}\psi $

\vspace*{3mm} for any $s\in S$ and $\Phi_s\cup\{\vp,\psi\}\subseteq Form_s$.
\end{center}
\end{theorem}

Let $\setmod{C}$ be a class of frames or a class of models and define
\begin{center}
$Taut(\setmod{C})=\{Taut(\setmod{C})_s\}_{s\in S}$, 

$Taut(\setmod{C})_s=\{\vp\in Form_s\mid {\mathcal M}\mos{s}\vp \mbox{  for any } {\mathcal M}\in\setmod{C}\}$
\end{center}
 We say that $\logl$ is {\em sound} with respect to $\setmod{C}$ if $\,\,\logl\subseteq Taut(\setmod{C})$.

\begin{proposition}\label{sound}(The soundness of the local deduction) Let $\setmod{C}$ be a class of frames or a class of models such that $\Lambda\subseteq Taut(\setmod{C})$.
If $s\in S$, $\Phi_s\subseteq Form_s$ and $\phi\in Form_s$  then $\Phi_s\vds{s}_{\logl}\phi$ implies $\Phi_s\mos{s}_{\setmod{C}}\phi$.
\end{proposition}

\begin{proof}
Let  $\mathcal{M}= (\mathcal{W}, \lbrace\mathcal{R}_{\srb}\rbrace_{\srb \in \Sigma}, V)$ be a model from $\setmod{C}$. 
We only prove that the generalization rule ($UG_\sigma^i$) is sound.
To prove this, we assume that $\srb\in \Sigma_{s_1\cdots s_n,s}$, $i\in [n]$ and
$\phi\in Form_{s_i}$ such that $\mathcal{M}, u \mos{s_i}_{\setmod{C}}\phi$  for any $u \in W_{s_i}$. Hence for  any 
$w\in W_s$ and \mbox{$(u_1,\ldots,u_n) \in W_{s_1}\times\cdots\times W_{s_n}$} such that $\mathcal{R}_{\srb} wu_1\ldots u_n$  we have $\mathcal{M},u_i \mos{s_i}_{\setmod{C}} \phi$, which means that \mbox{$\mathcal{M}, w\mos{s}_{\setmod{C}} \spt (\psi_{1}, \ldots ,\phi, \ldots \psi_{n})$} for any $\psi_1,\ldots, \psi_{i-1},\psi_{i+1},\ldots,\psi_n$ of appropriate sorts.
\end{proof} 

As a corollary we get: {\em the $(S,\Sigma)$-polyadic normal modal logic $\mathbf{K}$ is sound with respect to the class of all $(S,\Sigma)$-frames}.

\subsection{Canonical model. Completeness.}
Following closely the approach from \cite{mod}, in order to define the canonical models and to prove the completeness theorem, we need to study the consistent sets.

For any $s\in S$, we say that the set $\Phi_s \subseteq Form_s$ is ({\em locally}) $\logl$-{\em inconsistent}  if $\Phi_s\vds{s}_{\logl}\bot_s$ and it is ({\em locally})
$\logl$-{\em consistent} otherwise.

In the sequel by {\em consistent} we mean {\em locally consistent}.
{We analyze the  {\em global consistency} in Section \ref{secglobal}.} 

\begin{remark}
As in classical logic, one can easily prove the following.
\begin{itemize}
\item[(1)] A set  $\Phi_s \cup \{\vp\}_s \subseteq Form_s$ is $\logl$-inconsistent if and only if  $\Phi_s \vds{s}_{\logl} \neg \vp $.

\item[(2)] Assume $\logl$ is sound with respect to $\setmod{C}$. We further assume  $s\in S$ and  $\Phi_s\subseteq Form_s$  such that $\Phi_s$ is satisfiable on some model from $\setmod{C}$. Then $\Phi_s$ is $\logl$-consistent.
\end{itemize} 
\end{remark}

In the sequel we assume $\setmod{C}$ is a class of frames or a class of models.  

\medskip

We say that $\logl$ is {\em  complete} with respect to $\setmod{C}$ if $\,\, Taut(\setmod{C})\subseteq \logl$.  We say that 
$\logl$ is {\em strongly complete} with respect to $\setmod{C}$ if
\begin{center}
$\Phi_s\mos{s}_{\setmod{C}}\vp$  { implies } $\Phi_s\vds{s}_{\logl}\vp$ for any $s\in S$ and $\Phi_s\cup\{\vp\}\subseteq Form_s$.
\end{center}

\medskip

\noindent{The next result is a straightforward generalization of \cite[Proposition 4.12]{mod}.}

\begin{proposition}\label{helpcomp} The following are equivalent:\\
(i)  $\logl$ is strongly complete with respect to $\setmod{C}$,\\
(ii) any $\logl$-consistent set of formulas is satisfiable on some
model from $\setmod{C}$.
\end{proposition}

As usual, the {\em maximal consistent} sets of formulas are a main ingredient in defining the canonical model. In the local approach maximality is defined within a particular sort $s\in S$: a set of formulas $\Phi_s\subseteq Form_s$ is {\em maximal} $\logl${\em -consistent} if it is a maximal element in the set of all $\logl$-consistent sets of formulas of sort $s$ ordered by inclusion.

\begin{remark}\label{apartenmcs1}
In the mono-sorted setting, any maximal consistent set is closed to deduction. The same happens in our many-sorted approach: if $\Phi_s \subseteq Form_s$ is a maximal $\logl$-consistent set then 
\begin{center}
$\Phi_s  \vds{s}_{\logl} \vp$ iff $\vp \in \Phi_s$.
\end{center}
\end{remark}

The well-known properties of the maximal consistent sets hold in our setting, as well as the Lindenbaum's Lemma. For any sort $s\in S$, the proof is similar with the proof for the classical propositional logic, therefore we state them without proofs.

\begin{lemma}\label{16} 
If $s\in S$ and   $\Phi_s \subseteq Form_s$ is maximal $\logl$-consistent then the following properties hold for any $\vp, \psi \in Form_s$:
\begin{itemize}
\item[(i)] \label{16i}if $\vp$, $\vp \ri \psi \in \Phi_s$, then $\psi \in \Phi_s$
\item[(ii)] $\logl_s\subseteq \Phi_s$
\item[(iii)] \label{16iii}$\vp \in \Phi_s$ or $\neg \vp \in \Phi_s$
\item[(iv)] \label{16iiii}$\vp \wedge \psi \in \Phi_s$ if and only if  $\vp \in \Phi_s$ and   $\psi \in \Phi_s$.

\end{itemize}
\end{lemma}

\begin{lemma}(Lindenbaum's Lemma)
If $s\in S$ and $\Phi_s\subseteq Form_s$ is a $\logl$-consistent set of formulas then there is a maximal $\logl$-consistent set $\Phi_s^+$ such that $\Phi_s \subseteq \Phi_s^+$.
\end{lemma}

 In  order to prove that any consistent set is satisfiable, we define the {\em canonical model}.  Recall that $(S,\Sigma)$ is a fixed many-sorted signature, $P$ is an $S$-sorted set of propositional variables and our logic is $\logl$.

\begin{definition}\label{18}
The {\em canonical model} is $\mathcal{M}_{\logl}=(\mathcal{W}^{\logl}, \lbrace\mathcal{R}^{\logl}_{\srb}\rbrace_{{\srb} \in \Sigma}, V^{\logl})$ defined as follows:
\begin{itemize}
\item[(1)] for any $s\in S$,  $\mathcal{W}^{\logl}_s=\{ \Phi \subseteq Form_s \mid \Phi \mbox{ is maximal $\logl$-consistent} \}  $,
\item [(2)]\label{18ii} for any  $\srb \in \Sigma_{s_1 \ldots s_n,s}, w\in W^{\logl}_s, u_1\in W^{\logl}_{s_1},\ldots, u_n\in W^{\logl}_{s_n}$ we define
\begin{center}
$\mathcal{R}^{\logl}_{\srb} wu_1\ldots u_n$ iff $(\psi_1, \ldots, \psi_n)\in u_1\times\cdots\times u_n$ implies $\srb(\psi_1,\ldots , \psi_n)\in w$
\end{center}

\item[(3)] \label{18iii} $V^{\logl}= \lbrace V^{\logl}_s\rbrace_{s\in S}$ is the valuation defined by 
\begin{center}
$V^{\logl}_s(p) = \lbrace w \in W^{\logl}_s | p\in w \rbrace$ for any $s\in S$ and $p \in P_s$.
\end{center}
\end{itemize}

\end{definition}

\begin{lemma}\label{extruth}
If $s\in S$, $\phi\in Form_s$, $\sigma\in \Sigma_{s_1\cdots s_n,s}$ and $w\in W^{{\logl}}_s$ then the following hold:
\begin{itemize}

\item[(i)] $\mathcal{R}^{\logl}_{\srb} wu_1\ldots u_n$ iff for any  formulas $\psi_1,\ldots \psi_n$, $ \spt(\psi_1, \ldots, \psi_n) \in w$ implies $\psi_i \in u_i$ for some $i \in  [n] $.
\item[(ii)] If  $\tr( \psi_1, \ldots,\psi_n) \in w$ then for any $i \in  [n] $ there  is  $u_i \in W^{{\logl}}_{s_i}$  such that  
$\psi_1 \in u_1$, $\ldots$, $\psi_n \in u_n$ and $\mathcal{R}^{{\logl}}_\tr wu_1\ldots u_n$. 
\item[(iii)] $\mathcal{M}^{\logl}, w \mos{s} \phi$ if and only if $\phi \in w$.
\end{itemize}
\end{lemma}
\begin{proof} Note that (ii) from the above result is the analogue of the  {\em Existence Lemma} and (iii) is the analogue of the  {\em Truth Lemma} from \cite[Chapter 4.2]{mod}.

Let  $s\in S$, $\phi\in Form_s$, $\sigma\in \Sigma_{s_1\cdots s_n,s}$ and $w\in W^{{\logl}}_s$.
\begin{itemize}
\item[(i)] Assume that $\srb \in \si{s}{n}$ and let $\psi_1,\ldots,\psi_n$ be formulas of sorts $s_1,\ldots, s_n$, respectively.
 Suppose that $\mathcal{R}^{\logl}_{\srb} wu_1\ldots u_n$ and $ \spt(\psi_1, \ldots, \psi_n) \in w$. Assume that $\psi_i \not  \in u_i$ for all $i \in  [n] $.
Note that $u_i$ is a maximal consistent set of sort $s_i$ for any $i \in  [n] $.  By Lemma \ref{16}  we get $\neg \psi_i \in u_i$ for all $i \in  [n] $, which means that $\srb (\neg \psi_1, \ldots, \neg \psi_n)\in w$. Using again Lemma \ref{16} it follows that $\neg \srb (\neg \psi_1, \ldots, \neg \psi_n)\not \in w$, so $\spt ( \psi_1, \ldots,\psi_n)\not \in w$, contradiction with the hypothesis.

For the converse implication,  let \mbox{$(\psi_1, \ldots, \psi_n)\in u_1\times\cdots\times u_n$}, and assume that $\srb(\psi_1,\ldots , \psi_n)\not\in w$. Using Lemma \ref{16}, we infer that $\neg \srb(\psi_1,\ldots , \psi_n) \in w$, so $\spt(\neg\psi_1,\ldots , \neg\psi_n)\in w$. From hypothesis, there is $i\in  [n] $ such that $\neg \psi_i \in u_i$. Hence, for some $i\in  [n] $, $\psi_i \in u_i$ and $\neg \psi_i \in u_i$, which contradicts the fact that $u_i$ is consistent. Consequently, \mbox{$(\psi_1, \ldots, \psi_n)\in u_1\times\cdots\times u_n$} implies  $\srb(\psi_1,\ldots , \psi_n)\in w$, so $\mathcal{R}^{\logl}_{\srb} wu_1\ldots u_n$.

\item[(ii)] The proof is similar with \cite[Lemma 4.26]{mod}.
\item[(iii)]We make the proof by structural induction on $\phi$.
\begin{itemize}
\item[-)] $\mathcal{M}^{\logl}, w \mos{s} p$ iff  $w\in V^{\logl}_s(p)$ iff $p\in w$;

\item[-)]  $\mathcal{M}^{\logl},w \mos{s} \neg \phi$, if and only if $\mathcal{M}^{\logl},w \not\mos{s}\phi$ iff $\phi \not\in w$ (inductive hypothesis) iff $\neg \phi \in w $ (maximal ${\logl}$-consistent set Proposition \ref{16});
\item[-)]  $\mathcal{M}^{\logl},w \mos{s} \phi \vee \psi$ iff $\mathcal{M}^{\logl},w \mos{s} \phi$ or $\mathcal{M}^{\logl},w \mos{s} \psi$  iff $\phi \in w$ or $\psi \in w$ (inductive hypothesis) iff $\phi \vee \psi \in w$;
\item[-)] let $\tr \in \Sigma_{s_1 \ldots s_n,s}$ and $\phi=\tr (\phi_1, \ldots , \phi_n )$;
then  $\mathcal{M}^{\logl}, w \mos{s} \tr (\phi_1, \ldots , \phi_n )$ if and only if for any $i \in  [n] $ there exists $w_i \in W^{\logl}{s_i}$ such that $\mathcal{M}^{\logl},w_i  \mos{s_i} \phi_i$ and  $\mathcal{R}^{\logl}_{\tr} ww_1\ldots w_n$. Using the induction hypothesis we get $\phi_i\in w_i$ for any  $i \in  [n] $. Since $\mathcal{R}^{\logl}_{\tr} ww_1\ldots w_n$,  by definition we infer that $\phi\in w$. Conversly, suppose $\tr (\phi_1, \ldots , \phi_n ) \in w$. Using (ii),  for any $i \in  [n] $ there  is  $u_i \in W^{\logl}_{s_i}$  such that  
$\phi_1 \in u_1$, $\ldots$, $\phi_n \in u_n$ and $\mathcal{R}^{\logl}_\tr wu_1\ldots u_n$.  Using the induction hypothesis we get  $\mathcal{M}^{\logl}, u_i \mos{s} \phi_i$ for any $i \in  [n] $, so $\mathcal{M}^{\logl}, w\mos{s} \phi$. 
\end{itemize}
\end{itemize}
\end{proof}

\medskip

We  can now prove that the  local deduction is complete.

\begin{theorem}(Canonical model theorem)\label{corhelp}
For any $s\in S$, if $\Phi\subseteq Form_s$ is $\logl$-consistent then it is satisfiable over the canonical model $\mathcal{M}^{\logl}$. 
\end{theorem}
\begin{proof}
Let $s\in S$,  $\Phi\subseteq Form_s$ a $\logl$-consistent set and $w\subseteq Form_s$ a maximal $\logl$-consistent set such that $\Phi\subseteq w$. Hence  $\mathcal{M}^{\logl}, w \mos{s}_{\logl} \phi$ for any $\phi \in \Phi$, so $\Phi$ is satisfiable over the canonical model. 
\end{proof}

The above result asserts that, for any $\Lambda\subseteq Form$ the normal modal logic $\logl$ is complete with respect to the canonical model.

\begin{theorem}(Completeness of $\mathbf K$)
The $(S,\Sigma)$-polyadic normal modal logic $\mathbf K$ is strongly complete with respect to the class of all $(S,\Sigma)$-frames, i.e. 
for any $s\in S$, $\phi\in Form_s$ and $\Phi_s\subseteq Form_s$, 

\begin{center}
 $\Phi_s\vds{s}_{\mathbf K} \phi$ if and only if $\Phi_s \mos{s} \phi$.
\end{center}
\end{theorem}
\begin{proof}
It follows by Theorem \ref{corhelp} and Proposition \ref{helpcomp}. 
\end{proof}

\section{Global deduction for $\mathcal{ML}_S$}\label{secglobal}

In this section we study the global deduction that is especially relevant in our setting: in this case the set of hypothesis is an arbitrary $S$-sorted set. Note that for local deduction we use ${\vds{s}}$ and $\mos{s}$ when the set of hypotheses and the conclusion are of sort $s\in S$, while for global deduction we use $\vdash $ and $\mos{}$ meaning that the set of hypothesis may be an arbitrary $S$-sorted set (the sort of conclusion is uniquely determined in a particular context). We study the global deduction from a syntactic and semantic point of view, we prove a completeness theorem  and a general form of the deduction theorem. 

  In the following $\Lambda, \Gamma\subseteq Form$ are $S$-sorted sets of formulas: $\Lambda$ is the set of axioms and we study the deduction from $\Gamma$ in $\logl$. 

\begin{definition} 
If $\mathcal M$ is an $(S,\Sigma)$-model such that ${\mathcal M}\mos{s}\Gamma_s$ for any $s\in S$  then we say that  $\mathcal M$ is a {\em model} for $\Gamma$, and we write ${\mathcal M}\mos{}\Gamma$.

Let $s\in S$ and $\phi\in Form_s$. We say that $\phi$ is a {\em global semantic consequence of} $\Gamma$ in $\logl$,  and we write $\Gamma\mos{}_{\logl}\phi$,  if $\mathcal{M}\mos{} \Gamma$ implies  $\mathcal{M}\mos{s} \phi$ for any model $\mathcal{M}$  such that 
${\mathcal M}\mos{}\logl$.

Let $s\in S$ and $\phi\in Form_s$. We say that $\phi$ is a {\em global sintactic consequence of} $\Gamma$ in $\logl$,  and we write $\Gamma{\vds{}}_{\logl}\phi$, if there exists a sequence $\phi_1,\ldots, \phi_n$ such that $\phi_n=\phi$ and, for any $i \in  [n] $,  $\phi_i\in Form_{s_i}$ is an axiom or $\phi_i\in \Gamma_{s_i}$ or  it is inferred from $\phi_1,\ldots, \phi_{i-1}$ using  (MP$_{s_i}$) and (UG$_\sigma^k$) for some $\sigma\in\Sigma_{t_1\ldots t_{m_i},s_i}$ and $k\in [m_i]$.
\end{definition}

One can easily see that the global deduction is {\em sound}:

\begin{proposition} For any $\Gamma\subseteq Form$,  $s\in S$ and $\phi\in Form_s$ 
\begin{center}
$\Gamma{\vds{}}_{\logl}\phi$ implies $\Gamma \mos{}_{\logl} \phi$
\end{center} 
 \end{proposition}

Inspired by similar results in the mono-sorted setting ({see e.g. \cite[Chapter 3.1]{kracht}}) we analyze the relation between local and global deduction from syntactic point of view. 
Let $\Gamma\subseteq Form$ be an $S$-sorted set of formulas. We denote $\Gamma_{\rm G}=\bigcup_k\Gamma^k$ where $\{\Gamma^k\}_k$ is an increasing sequence of $S$-sorted sets of formulas, defined as follows:
\begin{center}
$\Gamma^0=\Gamma$,  $\,\,\,\Gamma^{k+1}_s=\Gamma^k_s\cup\left\{\sigma^{\mb}(\psi_{s_1},\ldots,\psi_{s_{i-1}},\gamma,
          \psi_{s_{i+1}}\ldots, \psi_{s_n})\mid i\in  [n] ,\right.$ 
          
 \hspace*{7cm}         $\left. \sigma\in\Sigma_{s_1\cdots s_n, s},  \gamma\in \Gamma^k_{s_i}\right\}$.
 \end{center}

\medskip
          
 \begin{proposition}\label{helpded}
 If $\phi\in Form_s$ for some $s\in S$ and $\Gamma\subseteq Form$,
 \begin{center}
 $\Gamma{\vds{}}_{\logl}\phi$ \mbox{ iff } ${\Gamma_{\rm G}}_s{\vds{s}}_{\logl}\,\phi$.
 \end{center}
 \end{proposition}         
\begin{proof} 
 All deductions are in $\logl$, so we simply write $\vds{s}$ and ${\vds{}}$ for local and global deduction, respectively.
Assume that $\Gamma\vds{}\phi$ and let $\phi_1,\ldots, \phi_n$ be a global proof of $\phi$ from $\Gamma$. We prove 
that ${\Gamma_{\rm G}}_s\vds{s}\phi_i$ by induction on $i\leq n$. 

If  $i=1$ then $\phi\in \Gamma_s$. For the induction step, we  analyze only the case when $\phi_i$ is derived using the deduction rules:
 
\noindent {\it Case 1}. if  $\phi_j=\phi_k\to\phi_i$ for some $j,k<i$ then   
${\Gamma_{\rm G}}_s\vds{s}\phi_k$ and 
${\Gamma_{\rm G}}_s\vds{s}\phi_j$; since the local deduction is closed to modus ponens we get ${\Gamma_{\rm G}}_s\vds{s}\phi_i$;

\noindent {\it Case 2}.  $\phi_i =\sigma^{\mb}(\psi_1,\ldots,\psi_{l-1},\phi_j,\psi_{l+1}\ldots, \psi_m)$ for some $j<i$, $\sigma\in \Sigma$ and 
$\psi_1,\ldots, \psi_m$ of appropriate sorts. 
Using the induction hypothesis ${\Gamma_{\rm G}}_{s_j}{\vds{s_j}}\phi_j$, so there are $\gamma_{1},\ldots, \gamma_{k}\in {\Gamma_{\rm G}}_{s_j}$ such that $\vds{s_j}\gamma_{1}\wedge\ldots\wedge\gamma_{k}\to\phi_j$. We infer that  $\sigma^{\mb}(\psi_{s_1},\ldots,\psi_{s_{i-1}},\gamma_{k'},\psi_{s_{i+1}}\ldots, \psi_{s_p})\in {\Gamma_{\rm G}}_{s}$
for any $k'\in [k]$. By Proposition \ref{genprop}(ii) we get 
${\Gamma_{\rm G}}_{s}{\vds{s}}\sigma^{\mb}(\psi_{1},\ldots,\psi_{{l-1}},\phi_j,\psi_{s_{l+1}}\ldots, \psi_{m})$, so
\mbox{${\Gamma_{\rm G}}_{s}\vds{s}\phi$}.

The other implication is obvious. 
\end{proof}

As a direct consequence we get the following. 

\begin{corollary}(Global deduction theorem I) If $\varphi,\psi\in Form_s$ for some $s\in S$ and $\Gamma\subseteq Form$ then 
$\Gamma\vds{}_{\logl}\varphi\to \psi$ iff 
${\Gamma_{\rm G}}_s\cup\{\varphi\}{\vds{s}}_{\logl}\psi$.
\end{corollary}

In the sequel we state another form of the deduction theorem. To do this we introduce further notations: for any $s\in S$ and $\varphi\in Form_s$ we define
 $\{\varphi\}^S$ by 
$\{\varphi\}^S_s=\{\varphi\}$ and 
$\{\varphi\}^S_t=\emptyset $ for $t\in S\setminus \{s\}$. 
Moreover,   $\{\varphi\}_G$ is $(\{\varphi\}^S)_G$.

\begin{theorem}(Global deduction theorem II)\label{thded2}
If $\varphi,\phi\in Form_s$ for some $s\in S$ and $\Gamma\subseteq
Form$ then the following are equivalent:\\
(i) $\Gamma\cup \{\varphi\}^S{\vds{}}_{\logl}\phi$\\
(ii) $\Gamma\vds{}_{\logl}\vp_1\wedge\cdots\wedge \vp_n\to \phi$ for
some $\vp_1,\ldots,\vp_n\in {\{\varphi\}_G}_s$.

\end{theorem}
\begin{proof}
All deductions are in $\logl$, so we simply write ${\vds{s}}$  and  ${\vds{}}$ for local and global deduction, respectively.
We note that $\Gamma\cup \{\varphi\}^S{\vds{}}\{\vp\}_G$, so (ii) implies (i) is straightforward.

We assume now that $\Gamma \cup \{\varphi\}^S \vds{}\phi$, so there exists
$\gamma_1,\ldots, \gamma_m$ a global proof of $\phi$ from
$\Gamma\cup \{\varphi\}^S$. By induction on $i\in [m]$ we prove 
\begin{center}
$\Gamma\vds{}\vp_1\wedge\cdots\wedge\vp_n\to \gamma_i$ for some $\vp_1,\ldots,\vp_n\in {\{\varphi\}_G}_{s_i}$  
\end{center}
where $s_i$ is the sort of $\gamma_i$. For $m=1$ we consider two cases:

\noindent {\it Case 1}. if $\gamma_1 \in \Gamma$ or $\gamma_1$ is an
axiom then the conclusion is obvious;

\noindent {\it Case 2}. if $\gamma_1 \in \{\vp\}_G$ then
$\Gamma{\vds{}}\gamma_1\to\gamma_1$ is true.

\noindent
For the induction step, we only analyze the case when $\gamma_i$ is derived using the deduction rules:

\noindent {\it Case 3}. if $\gamma_k=\gamma_j\to\gamma_i$ with
$j,k<i$ then we can find $\vp_1,\ldots,\vp_n\in [\varphi]_G$ and $l\in [n]$ such that 
$\Gamma{\vds{}}\vp_1\wedge\cdots\wedge \vp_l\to \gamma_j$ and $\Gamma{\vds{}}\vp_l\wedge\cdots\wedge \vp_n\to \gamma_k$. Using Proposition \ref{helpded} we consider local deduction on the sort $s_i$ and we can use the theorems of classical propositional logic on $s_i$. It follows that that 
$\Gamma\vds{}\vp_1\wedge\cdots\wedge \vp_n\to \gamma_i$.

\noindent {\it Case 4}.  let $\gamma_i = \spt (\psi_1,
\ldots,\psi_{l-1}, \gamma_j,\psi_{l+1}, \ldots, \psi_m)$
 where  $j<i$, $\sigma\in \Sigma$ and 
$\psi_1,\ldots, \psi_m$ have appropriate sorts; using the induction hypothesis and the universal generalization we get
\begin{center}
 $\Gamma{\vds{}} \spt(\psi_1,
\ldots,\psi_{l-1},\vp_1\wedge\cdots\wedge \vp_n\to
\gamma_j,\psi_{l+1}, \ldots, \psi_m)$;
\end{center}
for some $\vp_1,\ldots,\vp_n\in {\{\varphi\}_G}_{s_j}$ where $s_j$ is the sort of $\gamma_j$.
Using Proposition \ref{helpded} we consider local deduction on the sort $s_i$ of $\gamma_i$, so we can use Proposition
\ref{genprop} and $(K_\sigma^l)$ in order to infer that
\begin{center}
 ${\Gamma_G}_{s_i}{\vds{s_i}}
\spt (\ldots,\psi_{i-1},\vp_1 ,\psi_{i+1}, \ldots)
\wedge \ldots \wedge \spt (\ldots,\psi_{i-1},\vp_n ,\psi_{i+1},
\ldots) \to \gamma_i$ 
\end{center}
Since $\vp_k\in {\{\varphi\}_G}_{s_j}$ we get
$\spt (\ldots,\psi_{i-1},\vp_k ,\psi_{i+1}, \ldots)\in {\Gamma_G}_{s_i}$ for any $k\in [n]$ and we use again Proposition \ref{helpded} to get the desired conclusion.
\end{proof}

In the sequel we make preliminary steps towards proving a  global completeness theorem. To do this we define the {\em global consistency} as follows:  an $S$-sorted set  $\Gamma\subset Form$  is {\em globally inconsistent} in $\logl$ if $\Gamma {\vds{}}_{\logl}\bot_s$ for some $s\in S$ and {\em globally consistent} otherwise.

\medskip 

The next definition is a straightforward  generalization of  \cite[Definition 2.5]{mod}.

\begin{definition}
If ${\mathcal M}=({\mathcal W}, {\mathcal R},\rho)$ and ${\mathcal M}'=({\mathcal W}', {\mathcal R}',\rho')$ are $(S,\Sigma)$-models we say that ${\mathcal M}'$ is a submodel of $\mathcal M$ if ${\mathcal W}\subseteq {\mathcal W}'$, $R'_\sigma$ is the restriction of $R_\sigma$ for any $\sigma\in \Sigma$ and $\rho'_s(p)=\rho_s(p)\cap W'_s$  for any $s\in S$ and $p\in P_s$. We say that the submodel ${\mathcal M}'$ is a {\em generated $(S,\Sigma)$-submodel} of ${\mathcal M}$ if for any $\sigma\in \Sigma_{s1\ldots s_n,s}$
\begin{center}
$w\in W'_s$ and $R_\sigma ww_1\ldots w_n$ implies 
$w_1\in W'_{s_1},\ldots, w_n\in W'_{s_n}$
\end{center}
If ${\mathcal X}\subseteq {\mathcal W}$ is an $S$-sorted subset, the  {\em submodel generated by ${\mathcal X}$} is the smallest (with respect to inclusion) generated submodel of $\mathcal M$ that includes $\mathcal X$.
\end{definition}

\begin{theorem}\label{gcomp}
 If the $S$ sorted set  $\Gamma\subseteq Form$ is globally consistent in $\logl$,  then there exists a model of $\logl$ that is also a model from $\Gamma$. 
 \end{theorem}
\begin{proof} 
We follow the proof of  \cite[Proposition 3.1.3]{kracht}.
Since $\Gamma$ is globally consistent in $\logl$, by Proposition \ref{helpded}, ${\Gamma_G}_s$ is $\logl$-consistent for any $s\in S$, so we define 
\begin{center}
$\mathcal{W}^\Gamma_s=\{ \Phi_s \subseteq Form_s \mid \Phi_s \mbox{ is maximal $\logl$-consistent},{\Gamma_G}_s\subseteq \Phi \}$. 
\end{center}
Let ${\mathcal M}^\Gamma$ be the submodel of the canonical model ${\mathcal M}^{\logl}$ generated by 
\mbox{${\mathcal W}^\Gamma=\{{\mathcal W}^\Gamma_s\}_{s\in S}$}.
It follows that ${\mathcal M}^\Gamma$ is a $\logl$-model and 
${\mathcal M}^\Gamma\mos{}\Gamma$.
\end{proof}

Note that the global consistency is particularly interesting in a many-sorted setting and it is affected by the way sorts are connected through operations. 

\section{ Algebraic semantics for ${\mathcal M}{\mathcal L}_S$}
\label{secalg}
In this section we add an algebraic perspective to our logic. Since the boolean algebras with operators (BAO) and the complex algebras are the structures of the polyadic modal logic, we define their many-sorted versions. Therefore, we introduce the \textit{many-sorted boolean algebras with operators} and we generalize  the J\'{o}nsson-Tarski Theorem. In doing this we follow closely \cite[Chapter 5]{mod}.

Recall that $(S,\Sigma)$ is a fixed many-sorted signature.

\begin{definition}\label{sbao}
 An {\em $(S,\Sigma)$-boolean algebra with operators} ({\em $(S,\Sigma)$-BAO}) is a structure
\begin{center} 
 $ \mathfrak{A} = (\lbrace {\mathcal A}_s\rbrace _{s\in S}, \lbrace f_\tr\rbrace_{\tr \in \Sigma}) $
 \end{center}
where $\mathcal{A}_s=(A_s, \vee_s, \neg_s, 0_s)$ is a boolean algebra for any sort $s\in S$ and, for any  $\sigma \in \Sigma_{s_1 \ldots s_n, s}$,   $f_\tr : A_{s_1} \times \ldots \times A_{s_n} \rightarrow A_s$ satisfies the following properties:\\
(N)  $f_\tr(a_1, \ldots, a_n) = 0_s$ whenever $a_i = 0_{s_i}$  for some $i \in  [n] $,\\
(A) $ f_\tr(a_1, \ldots,a_i \vee_{s_i} a'_{i}, \ldots, a_n) = 
f_\tr(a_1, \ldots,a_i, \ldots, a_n) \vee_s f_\tr(a_1, \ldots, a'_{i}, \ldots, a_n)$, for any $i\in [n]$.

\end{definition}

\noindent In the above definition, N stands for {\em Normality} and A stands for {\em Additivity}. The dual operators are defined by
$f_{\sigma^{\mb}}(a_1, \ldots, a_n):=\neg_s f_\sigma(\neg_{s_1}a_1, \ldots,\neg_{s_n} a_n)$ for any $a_1\in A_{s_1},\ldots, a_n\in A_{s_n}$.

An $(S,\Sigma)$-BAO can be equivalently defined as a many-sorted structure
\begin{center}
$ \mathfrak{A} = (\lbrace A_s\rbrace _{s\in S}, \lbrace \cup_s \rbrace_{s\in S}, \lbrace \neg _s\rbrace _{s\in S}, \lbrace 0_s \rbrace _{s\in S},\lbrace f_\tr\rbrace_{\tr \in \Sigma}) $
\end{center}
satisfying the equations of boolean algebras on any sort. The notions of subalgebra, congruence, homomorphism are defined as in universal algebra. Consequently,  if $h:\mathfrak{A}\to\mathfrak{B}$ is an $(S,\Sigma)$-BAO homomorphism then $h_s:{\mathcal A}_s\to {\mathcal B}_s$ are boolean algebra homomorphisms for any $s\in S$.  

In the sequel we shall omit the sort when the context is clear.

Our main examples are the {\em $(S,\Sigma)$-complex algebras} and the {\em Lindenbaum-Tarski algebra} of ${\mathcal M}{\mathcal L}_S$.

\begin{example}(Complex algebras)
Let  ${\mathcal F}=\left(\{ W_s \}_{s\in S}, \{\mathcal{R}_{\srb}\}_{\srb \in \Sigma}\right)$ be an $(S,\Sigma)$-frame. For any $s\in S$ let $\mathfrak{P}(W_s) =( \mathcal{P}(W_s),\cup_s, \neg_s, \emptyset_s)$ be the powerset algebra of $W_s$. Hence the {\em full complex $(S,\Sigma)$-algebra} determined by $\mathcal F$ is 
\begin{center}
${\mathfrak F}=(\lbrace \mathfrak{P}(W_s)\rbrace_{s\in S},\lbrace m_\sigma\rbrace_{\sigma \in \Sigma})$
\end{center}
where, for any $\sigma \in \Sigma_{s_1 \ldots s_n, s}$  we set
$m_\sigma: \mathcal{P}(W_{s_1}) \times \cdots \times \mathcal{P}(W_{s_n}) \rightarrow \mathcal{P}(W_{s})$
\begin{center}
 \mbox{$m_\sigma(X_1, \ldots, X_n)= \lbrace w \in W_s|\,
\mathcal{R}_\sigma ww_1\ldots w_n  \mbox{ for some } \  w_1 \in X_{1}, \ldots, w_n \in X_{n}\rbrace$}
\end{center}
\noindent and the dual $l_\sigma\left(X_1, \ldots, X_n \right):=\neg_s m_\sigma\left(\neg_{s_1}X_1, \ldots,\neg_{s_n} X_n \right)$ is defined by 
\begin{center}
 $w\in l_{\mathcal{R}_\sigma} \left(X_1, \ldots, X_n \right)$ iff $R_\sigma ww_1, \ldots, w_n $ implies $w_i \in X_i$ for some $\ i \in  [n] $.
\end{center}

\noindent One can easily see that $\mathfrak F$ is an
$(S,\Sigma)$-BAO. 
\end{example}

\begin{definition}
A {\em $(S,\Sigma)$-complex algebra} is a subalgebra of a full complex algebra determined by an $(S,\Sigma)$-frame. 
\end{definition}

Let $\mathfrak{A}=\left( \{\mathcal{A}_s\}_{s\in S}, \{f_\sigma\}_{\sigma \in \Sigma}\right)$ be an $(S,\Sigma)$-BAO. For any $s\in S$ let $Uf(A_s)$  be the set of all  ultrafilters  of ${\mathcal A}_s$. For any 
$\sigma\in \Sigma_{s_1\ldots s_n,s}$ and $w\in Uf(A_s),$ $ w_1\in Uf(A_{s_1}), \ldots , w_n\in Uf(A_{s_n})$ we define
\begin{center}
$ Q_{f_\sigma} ww_1\ldots w_n \mbox{ iff for all }  a_1 \in w_1, \ldots, a_n\in w_n, \ f_\sigma(a_1, \ldots, a_n) \in w.$
\end{center}
The $(S,\Sigma)$-frame ${\mathcal U}f(\mathfrak{A})=(\{Uf(A_s)\}_{s\in S}, \{Q_{f_\sigma}\}_{\sigma \in \Sigma} )$ is the {\em ultrafilter frame of $\mathfrak{A}$}. Denote ${\mathfrak F}({\mathfrak A})$ the {\em full complex algebra determined by ${\mathcal U}f(\mathfrak{A})$}.

\medskip

We are ready now to prove the generalization of the J\'{o}nsson-Tarski theorem to $(S,\Sigma)$-BAOs.

\begin{theorem}(J\'{o}nsson-Tarski theorem for $(S,\Sigma)$-BAOs)\label{jttheorem}
Any $(S,\Sigma)$-BAO $\mathfrak{A}$ can be embedded in ${\mathfrak F}({\mathfrak A})$, the full complex $(S,\Sigma)$-algebra of its ultrafilter $(S,\Sigma)$-frame.  
\end{theorem}
\begin{proof} 
We only sketch the proof, since it follows closely the proof for the mono-sorted setting (see \cite[Theorem 5.43]{mod}).  
Let $\mathfrak{A}=\left( \{\mathcal{A}_s\}_{s\in S}, \{f_\sigma\}_{\sigma \in \Sigma}\right)$ be an $(S,\Sigma)$-BAO and $W=\{W_s\}_{s\in S}$, where $W_s$ is the ultrafilter set of $\mathcal{A}_s$. Using the Stone Representation Theorem we get a morphism
 $r=\{r_s\}_{s\in S}$ such that $r_s:A_s \ri \mathcal{P}(W_{s})$ is defined  by $r_s(a)=\{ w \in W_s| a\in w \}$.
 Note that $r_s$ is a boolean  embedding for any $s\in S$. We have to prove that $r$ is an embedding of
 $\mathfrak{A}$ into the full complex algebra of its ultrafilter frame, i.e. we have to prove that $r$ is a modal homomorphism:
\[\label{homs}
r_s(f_{\sigma}(a_1, \ldots, a_n))=m_{Q_{f_{\sigma}}}(r_{s_1}(a_1), \ldots,r_{s_n}(a_n) ) \tag{\mbox{H} }
\]

\noindent Let $w \in W_s$  be an ultrafilter of $\mathcal{A}_s$. We have to prove that  $f_{\sigma}(a_1, \ldots, a_n ) \in w$ iff   there exist $w_1 \in W_{s_1},\ldots , w_n\in W_{s_n}$ such that $ a_i \in w_i$ for any $i\in [n]$ and $Q_{f_{\sigma}} ww_1 \ldots w_n$. To simplify the notation, we fix the operation $\sigma$ and we use the notation $f$ instead of $f_\sigma$.

We prove (H) by induction on the arity  of $f$. Assume that $f$ is an unary operator. Therefore, we show that\begin{center}
$r_{s}(f(a))=m_{Q_f}(r_{s_1}(a))$
\end{center}
If $w \in m_{Q_f}(r_{s_1}(a)) $ there exists  $w_1 \in r_{s_1}(a) $ such that  $Q_f w w_1$. It follows that $a\in w_1$ and  $Q_f w w_1$, so $f(a) \in w $.

Conversely, assume that $w$ is an ultrafilter such that $w\in r_s(f(a))$, i.e. $f(a)\in w$. We need to prove that $w \in m_{Q_f}(r_{s_1}(a))$, so we have to find an ultrafilter $w_1$ such that $a \in w_1$ and $Q_f w w_1 $.  If  $F_{s_1}:=\{v \in A_{s_1}|\neg f(\neg v)\in w  \}$ then we prove that $F_{s_1}$ is closed under taking meets. Let $v, t \in F_{s_1} \Rightarrow  \neg f(\neg v)\in w $  and $\neg f(\neg t)\in w $. By additivity of $f$, $ f(\neg v)\vee  f(\neg t)=  f(\neg v \vee \neg  t)=f(\neg (v \wedge t))$, hence $\neg f(\neg (v \wedge  t)) = \neg (f(\neg v)\vee  f(\neg t)) = \neg f(\neg v)\wedge  \neg f(\neg t)$. Since $w$ is an ultrafilter and  $\neg f(\neg v), \neg f(\neg t)\in w $, we can infer that $\neg f(\neg (v \wedge  t))\in w$. We proved that  $F_{s_1}$ is closed under taking meets. It follows that the filter generated by $F_{s_1}$ is proper, so there is an ultrafilter $w_1$ such that $F_{s_1} \subseteq w_1$. Note that $\neg f(\neg v) \in w$ implies  $v \in w_1$, so  $Q_f w w_1$ it holds.

Assume that induction hypothesis (H) holds for $n$ and let $f$ de a normal and additive function or arity $n+1$. One can further proceed as in \cite[Theorem 5.43]{mod}.
\end{proof}

Finally, we prove the algebraic completeness of for systems of many-sorted modal logic. 

\begin{example}(The Lindenbaum-Tarski algebra of $\logl$)
For any $s\in S$ the theorems of classical propositional logic are in $\logl_s$ for any $s\in S$ and the boolean connectives $\vee$ and $\neg$ preserve the sort,  so
${\mathcal F}orm_s=(Form_s,\vee,\neg,\bot_s)$ is a boolean algebra.
We define $\ \ f_\sigma (\phi_1, \ldots, \phi_n):=\sigma(\phi_1, \ldots, \phi_n) $ for any $ \ \ \ \ \ \ \ \sigma\in \Sigma_{s_1\ldots s_n, s}$ and $\phi_1\in Form_{s_1},\ldots, \phi_n\in Form_{s_n}$. One can easily see that $\mathfrak{Form}= ( \{{\mathcal F}orm_s\}_{s \in S}, \lbrace f_\tr\rbrace_{\tr \in \Sigma})$ is an $(S,\Sigma)$-BAO.

Let $\logl$ be an $(S,\Sigma)$-polyadic normal modal logic defined as in Section \ref{mainlogded}. As usual, $\logl$ determines an equivalence relation on formulas. In our setting this relation is $S$-sorted:  $\equiv _{\logl}=\{\equiv_{\logl}^s\}_{s\in S}$  where 
\begin{center}
$\phi \equiv ^{s}_{\logl} \psi $ iff    $\vds{s}_{\logl}  \phi \leftrightarrow \psi $ for any $\phi,\psi\in Form_s$.
\end{center}
By Proposition \ref{genprop} (iv),  $\equiv_{\logl}$ is a congruence relation on $ \mathfrak{Form}$. We denote ${\mathcal L}_{\logl}^s :={\mathcal Form_s} \equiv_{\logl}^s$ for any $s\in S$ and 
\begin{center}
$ \mathfrak{L}_{\logl} = (\lbrace\mathcal{L}_{\logl}^s \rbrace_{s\in S},\{\tilde{f}\}_\tr\rbrace_{\tr \in \Sigma } )$
 \end{center}
We denote $[\vp]_{\logl}$ the class of $\vp$ determined by $\equiv_{\logl}^s$ (we simply write $[\vp]$ when $\logl$ is fixed). Hence 
$\tilde{f}_\sigma([\phi_1], \ldots, [\phi_n]):= [f_\sigma(\phi_1, \ldots, \phi_n)]= [\sigma(\phi_1, \ldots, \phi_n)]$
for any $\sigma\in \Sigma_{s_1\ldots s_n, s}$ and $\phi_1\in Form_{s_1},\ldots, \phi_n\in Form_{s_n}$. Note that $\mathfrak{L}_{\logl}$ satisfies (N) and (A) by Remark \ref{alt}.

The $(S,\Sigma)$-BAO $\mathfrak{L}_{\logl}$ is the {\em  Lindenbaum-Tarski algebra} of the $(S,\Sigma)$-polyadic normal modal logic $\logl$.
\end{example}

Our next goal is to prove the algebraic completeness for $\logl$.
Recall that $P=\{P_s\}_{s\in S}$ is the set of propositional variables. If $\mathfrak{A}=\left( \{\mathcal{A}_s\}_{s\in S}, \{f_\sigma\}_{\sigma \in \Sigma}\right)$ is an $(S,\Sigma)$-BAO, then an 
{\em assignment in $\mathfrak{A}$} is an $S$-sorted function $e=\{e_s\}_{s\in S}$, $e_s:P_s\to A_s$ for any $s\in S$. As usual, any assignment can be uniquely extended to ${\mathfrak e}:\mathfrak{Form}\to {\mathfrak A}$, a homomorphism of $(S,\Sigma)$-BAOs. We define 
$Taut(\mathfrak{A})=\{Taut(\mathfrak{A})_s\}_{s\in S}$ where 
\begin{center}
$Taut(\mathfrak{A})_s=\{\phi\in Form_s\mid \mathfrak{e}_s(\phi)=1_s \mbox{ in } A_s \mbox{ for any  assignment }e_s:P_s\to A_s \}$
\end{center}
 
One can easily prove  that $\logl$ is {\em sound with respect to assignments} in $(S,\Sigma)$-BAOs. Moreover, the following holds.  
 
\begin{proposition}\label{ismod} $Taut(\mathfrak{L}_{\logl})_s=
\{\phi\mid \,\,\phi\in Form_s, \,\, \vds{s}_{\logl}\phi\}$ for any $s\in S$.
\end{proposition}
\begin{proof} It is a consequence of the fact that,  by Definition \ref{clsub},  $\logl$ is closed to $S$-sorted uniform substitutions.  

If $s\in S$ and $\phi\in Taut(\mathfrak{L}_{\logl})_s$ then $\mathfrak{e}(\phi)=[\top_s]$ for any assignment $e$ in  $\mathfrak{L}_{\logl}$. In particular, $[\phi]=[\top_s]$, so $\vds{s}_{\logl}\phi\leftrightarrow \top_s$, i.e.  $\vds{s}_{\logl}\phi$. Conversely, assume that $\vds{s}_{\logl}\phi$ and let $e$ be an assignment in $\mathfrak{L}_{\logl}$. For any $t\in S$ and propositional variable $p\in P_t$ there is $\psi_p\in Form_t$ such that $e_t(p)=[\psi_p]$. We can prove by structural induction that for any  $t\in S$ and $\psi\in Form_t$,  $\mathfrak{e}_t(\psi)=[\phi']$, where $\psi'$ is the formula obtained from $\psi$ by uniform substitution, replacing any propositional variable $p$ by $\psi_p$.  By Definition \ref{clsub} we infer that $\logl \subseteq Taut(\mathfrak{L}_{\logl})$. To finish the proof we note  that the deduction rules are sound  with respect to assignments in the Lindenbaum-Tarski algebra $\mathfrak{L}_{\logl}$.
\end{proof}

\begin{lemma}\label{free}
If $\logl\subseteq Taut(\mathfrak{A})$ then any assignment $e_s:P_s\to A_s$ for any $\ \ s\in S$ can be uniquely extended to a  homomorphism of $(S,\Sigma)$-BAOs $\tilde{{\mathfrak e}}:\mathfrak{L}_{\logl}\to {\mathfrak A}$ such that 
$\tilde{{\mathfrak e}}_s([\phi])=\mathfrak{e}_s(\phi)$ for any $s\in S$ and $\phi\in Form_s$.
\end{lemma}
\begin{proof} Since $\logl\subseteq Taut(\mathfrak{A})$ we get 
$\equiv_{\logl}\subseteq Ker(\mathfrak{e})$, so we apply the universal property of quotients. 
\end{proof}

\begin{theorem}(Algebraic completeness for $\logl$)\label{alg1}
With the above notations, the following are equivalent for any 
$s\in  S$ and $\phi\in Form_s$:
\begin{itemize}
\item[(i)] $\vds{s}_{\logl}\phi$
\item[(ii)] $[\phi]_{\logl}=[\top_s]_{\logl}$ in $\mathfrak{L}_{\logl}$
\item[(iii)] $\mathfrak{e}_s(\phi)=1_s$ in $A_s$ for any $(S,\Sigma)$-BAO  $\mathfrak{A}$ such that $\logl\subseteq Taut(\mathfrak{A})$ and any assignment $e$ in $\mathfrak{A}$.
\end{itemize}
\end{theorem}
\begin{proof}
$(i)\Leftrightarrow (ii)$ We have $\vds{s}_{\logl}\phi$ iff $\phi\equiv_{\logl}^s\top_s$ for any $s\in S$.\\
$(ii)\Leftrightarrow (iii)$ One implication follows by Proposition \ref{ismod}.
For the other, assume $\mathfrak{A}$ is an $(S,\Sigma)$-BAO and $e$ is an assignment in $\mathfrak{A}$. Let $\tilde{{\mathfrak e}}:\mathfrak{L}_{\logl}\to {\mathfrak A}$ be the homomorphism defined in Lemma \ref{free}. Hence $\mathfrak{e}_s(\phi)= \tilde{{\mathfrak e}}_s([\phi])= \tilde{{\mathfrak e}}_s([\top_s])=1_s$ in $A_s$. Note that, for the last equality, we used the fact that  $\tilde{{\mathfrak e}}_s$  preserves the top element.   
\end{proof}

\begin{corollary}(Algebraic completeness for $\mathbf K$)\label{alg2}
For any $s\in S$ and $\phi\in Form_s$ $\vds{s}_{\mathbf K}\phi$ if and only if $\mathfrak{e}_s(\phi)=W_s$  for any $(S,\Sigma)$-frame 
${\mathcal F}=\left(\{ W_s \}_{s\in S}, \{\mathcal{R}_{\srb}\}_{\srb \in \Sigma}\right)$ and any assignment  $e$ in the full complex algebra $\mathfrak{F}$. 
\end{corollary}
\begin{proof} It follows by Theorems \ref{alg1}, \ref{jttheorem} and the fact that ${\mathbf K}\subseteq Taut(\mathfrak{F})$ for any full complex algebra ${\mathfrak F}$.
\end{proof}


\section{Related logical systems}\label{exlog}

Traditionally, program verification within \textit{modal logic}, as showcased by \textit{dynamic logic} \cite{dynamic},  is following the mainstream axiomatic approach proposed by Hoare/Floyd \cite{hoare,floyd}. More recently, Ro\c su \cite{rosu} proposed \textit{matching logic} and \textit{reachability logic} as an alternative way to prove program correctness, using directly the (executable) operational semantics of a language.  The completeness theorem for matching logic is proved using a interpretation in the first-order logic with equality. The starting point of our investigation was the representation of the  (mono-sorted) polyadic modal logic  as a particular system of matching logic in \cite[Section 8]{rosu}. Our initial goals were: to understand the propositional part of matching logic, to give a self-contained proof of the completeness theorem, to identify the algebraic theory and to  investigate  the relation with modal logic. The many-sorted system ${\mathbf K}$ is an initial step in this direction. 

{\em Matching logic}, defined in \cite{rosu}, is a many-sorted first-order logic for program specification and verification
(see Appendix B for a brief presentation).  The formulas of Matching logic are called {\em patterns}.  Semantically, a pattern is interpreted as a set of elements, with the restriction that variables are interpreted as singletons. A sound and complete system of axioms is defined in \cite{rosu}, the completeness being proved by translation in  pure first-order logic with equality. As shown in \cite{rosu}, classical propositional logic, pure predicate logic, separation logic can be seen as instances of Matching logic. In \cite[Section 8]{rosu} the (mono-sorted) polyadic modal logic is represented as a particular system of Matching logic. This was the starting point of our investigation. 

Let  $(S,\Sigma)$ be  a many-sorted signature. A {\em Matching logic $(S,\Sigma)$-model} is 
$M=(\{M_s\}_{s\in S},\{\sigma_M\}_{\sigma\in\Sigma})$ where 
$\sigma_M : M_{s_1}\times \ldots \times M_{s_n} \rightarrow \mathcal{P}(M_s)$ for any  $ \sigma \in \Sigma_{s_1 \ldots s_n,s}$.  For $\sigma\in \Sigma_{s_1 \ldots s_n,s}$ we define \mbox{$R_\sigma\subseteq M_{s_1}\times\cdots\times M_{s_n}$} by $R_\sigma w w_1\ldots w_n$ iff $w\in \sigma_M(w_1,\ldots, w_n)$. Hence ${\mathcal M}=(\{M_s\}_{s\in S}, \{R_\sigma\}_{\sigma\in \Sigma})$ is an $(S,\Sigma)$-frame so we consider ${\mathfrak M}=(\{{\mathfrak P}(M_s)\}_{s\in S},\{m_\sigma\}_{\sigma\in \Sigma})$ the full complex algebra determined by $\mathcal M$. One can easily see that $m_\sigma =\tilde{\sigma}_M$ for any $\sigma\in \Sigma$, where $\tilde{\sigma}_M$ is uniquely defined by $\sigma_M$ (note that in \cite[Definition 2.2]{rosu} $\tilde{\sigma}_M$ is identified with $\sigma_M$). 

\begin{remark}\label{alg3}
To any Matching logic model we can associate an $(S,\Sigma)$-frame and vice versa.  
\end{remark}

Note that the variables of Matching logic are always interpreted as singletons, while our propositional variables have arbitrary values. 
Let $(S,\Sigma_P)$ be the signature obtained by adding a constant operation  symbol  for any propositional variables $p\in P_s$ and for any $s\in S$. Hence any formula in our $(S,\Sigma)$-polyadic modal logic is a formula in Matching logic over $(S,\Sigma_P)$. 
 
The relation between the many-sorted polyadic modal logic developed in this paper and Matching logic can be summarized as follows:
\begin{itemize}
\setlength{\itemsep}{1pt}
  \setlength{\parskip}{0pt}
  \setlength{\parsep}{0pt}
\item  the system $\mathbf K$ presented in this paper can be seen as the {\em propositional counterpart of Matching logic},
\item  the system $\mathbf K$ offers direct proofs for various types of completeness results,
\item by Remark \ref{alg3} and Corollary \ref{alg2}, the algebraic theory of Matching logic is the theory of 
 many-sorted Boolean algebras with operators, 
\item   { the variables of Matching logic are similar with the {\em state variables} in hybrid modal logic, so our conjecture is that {\em Matching logic can be presented as a hybridized many-sorted polyadic modal logic}. This will be the subject of further investigations.} 
\end{itemize}

Even if Matching logic was the starting point of our research, one main issue was to connect our logic with already existing systems of many-sorted modal logic.

\begin{example} In \cite{manys2} the author defines a sound and complete two-sorted modal logic for projective planes. If we set $S=\{p,l\}$ ($p$ for points, $l$ for lines), $\Sigma_{l,p}=\{\langle 01 \rangle\}$  and $\Sigma_{p,l}=\{\langle 10 \rangle\}$, $[01]:=\langle 01\rangle ^{\mb}$ and $[10]:=\langle 10\rangle ^{\mb}$ then our modal language ${\mathcal M}{\mathcal L}_{(S,\Sigma)}$ is the modal language $MLG_2$ defined in \cite{manys2}. Moreover, if $\Lambda=\{CJ, D_l, D_p,4_{\langle .\rangle}, 4_{\langle -\rangle}\}$ from \cite[Definition 4.1]{manys2} then our system $\logl$ is equivalent with the system AXP from \cite{manys2}. Consequently, all our general results stand for AXP.
\end{example}

\begin{example}
The {\em Propositional Dynamic logic with Arrows} defined in  \cite[Section 4]{manys3} is a two-sorted modal logic. Let $S=\{p,s\}$ ($p$ for programs, $s$ for statements), $\Sigma_{p,s}=\{{\scriptstyle <\cdot >}\}$, $\Sigma_{s,p}=\{?\}$, $\Sigma_{pp,p}=\{{\scriptscriptstyle{^\bullet}}\}$, 
$\Sigma_{p,p}=\{^\otimes\}$, $\Sigma_{\lambda,p}=\{Id\}$, $\Sigma_{p,p}=\{^*\}$. Here $\scriptscriptstyle{^\bullet}$ denotes the composition, $^\otimes$ denotes the converse and $Id$ is the identity, so on the sort $p$ one can define a system of arrow logic. 
The operation $^*$ is the program iteration, while  $\scriptstyle <\cdot >$ and  $?$ connect programs and statements. In \cite{manys3} the following deduction rule is considered:
{\em if $\phi\to\alpha$ and $\alpha{\scriptscriptstyle{^\bullet}}\alpha\to\alpha$ are provable, then $\phi^*\to\alpha$ is provable.}
For the moment, our many-sorted modal logic has  only {\em modus ponens} and {\em universal generalization} as deduction rules so one should study what happens when new deduction rules are added. However we notice that using $\Delta, L, R$ instead on $\scriptstyle <\cdot >,\,\,  ?$ as suggested in \cite[Section 4]{manys3} one can replace the above deduction rule with the induction axioms for  iteration \cite[Section 5.5]{dynamic}. Even if more is needed for a  complete development, such a system would fit in our setting and would benefit of our general results. 
 \end{example}

\begin{example} In \cite{contextModal}, a modal semantics is given for {\em Context logic} and {\em Bunched logic}, both being logics for reasoning about data structures. These logics are interpreted into a two-sorted polyadic normal modal logic. The sound and complete axiomatization given in \cite[Section 3]{contextModal} fit in our development so all our general results stand in this case.  
\end{example}

%
\section{A modal logic approach to operational semantics}\label{os}

In the sequel, we take first steps in exploring the amenability of dynamic logic in particular, and of modal logic in general, to express operational semantics of languages (as axioms), and to make use of such semantics in program verification. 

In the remainder of the paper we develop a particular system $\mathbf{K}\Lambda$ which uses dynamic logic in a many-sorted setting. Our goal is to express operational semantics of languages as axioms in this logic, and to make use of such semantics in program verification.
 As a first step towards that goal, we consider here the SMC Machine described by Plotkin \cite{plotkin}, we derive a Dynamic Logic set of axioms from its proposed transition semantics, and we argue that this set of axioms can be used to derive Hoare-like assertions regarding functional correctness of programs written in the SMC machine language.

The semantics of the \textit{SMC machine} as laid out by Plotkin consists of a set of transition rules defined between configurations of the form $\left\langle S, M, C \right\rangle$ where $S$ is a $ValueStack$ of intermediate results, $M$ represents the $Memory$, mapping program identifiers to values, and $C$ is the $ControlStack$ of commands representing the control flow of the program.
Inspired by the \textit{Propositional Dynamic Logic} (PDL) \cite{dynamic}, we identify the $ControlStack$ with the notion of ``programs'' in dynamic logic, and use the ``;'' operator from dynamic logic to denote stack composition.  We define our formulas to stand for {\em configurations} of the form $config(vs, mem)$
comprising only a value stack and a memory.
 Similarly to PDL, we use the modal operator
$ [\_]\_ : ControlStack \times Config \to Config$
to assert that a configuration formula must hold after executing the commands in the control stack.
The axioms defining the dynamic logic semantics of the SMC machine are then formulas of the form 
$cfg \to [ctrl] cfg'$ 
saying that a configuration satisfying $cfg$ must change to one satisfying $cfg'$ after executing $ctrl$.

\begin{figure}[!htb]
\small{

\noindent \begin{minipage}[t]{.45\textwidth}

{\bf Syntax}

\begin{verbatim}
 Nat ::=  natural numbers
 Var ::=  program variables
Bool ::= true | false
AExp ::= Nat | Var | AExp + AExp
BExp ::= AExp <= AExp
Stmt ::= x := AExp
        | if BExp 
          then Stmt 
          else Stmt
        | while BExp do Stmt
        | skip
        | Stmt ; Stmt
\end{verbatim}
\end{minipage}\hfill 
\begin{minipage}[t]{.45\textwidth}
{\bf Semantics}
\begin{verbatim}
     Val ::= Nat | Bool
ValStack ::= nil 
             | Val . ValStack
     Mem ::= empty | set(Mem, x, n) 
             | get(x,n)
CtrlStack ::= c(AExp)
             | c(BExp)  
             | c(Stmt)
             | asgn(x)   
             | plus   | leq    
             | Val ?
             | c1 ; c2
  Config ::=  config(ValStack, Mem)          

\end{verbatim}
\end{minipage}
}
\caption{\bf Signature}
\end{figure}

In Fig. 1  we introduce the signature of our logic as a context-free grammar (CFG) in a BNF-like form.
We make use of the established equivalence between CFGs and algebraic signatures (see, e.g., \cite{HHKR89}), mapping non-terminals to sorts and CFG productions to operation symbols.
Note that, due to non-terminal renamings (e.g., Exp ::= Int), it may seem that our syntax relies on subsorting.  However, this is done for readability reasons only.  The renaming of non-terminals in syntax can be thought of as syntactic sugar for defining injection functions.  For example, Exp ::= Int can be thought of as Exp ::= int2Exp(Int), and all occurrences of an integer term in a context in which an expression is expected could be wrapped by the int2Exp function.

The sorts \texttt{CtrlStack} and \texttt{Config} correspond to "programs" and "formulas" from PDL, respectively. 
  {Therefore the usual operations of dynamic logic $;$, $\cup$, $^*$,  $[\_]$ are defined accordingly \cite[Chapter 5]{dynamic}. We depart from PDL with the definition of $?$ (test): in our setting, in order to take a decision, we test the top value of the value stack. Consequently,   the signature of the test operator is $?:Val\to CtrlStack$.}

We are ready to define our axioms. For the rest of the paper, whenever $phi$ is a theorem of sort $s$, i.e. $\vds{s}\phi$, we will simply write $\vdash\phi$, since the sort $s$ can be easily inferred.

The first group of axioms is inspired by the axioms of PDL \cite[Chapter 5.5]{dynamic}.
In the following, $\pi$, $\pi'$ are formulas of sort \texttt{CtrlStack} ("programs"), $\gamma$ is a  formula of sort \texttt{Config} (the analogue of "formulas" from PDL), $v$ and $v'$ are variables of sort \texttt{Var}, $vs$ has the sort \texttt{ValStack} and $mem$ has the sort \texttt{Mem}.

\medskip
\noindent{\bf PDL-inspired axioms}\\
$\begin{array}{ll}
(A\cup) & [\pi \cup \pi'] \gamma \leftrightarrow [\pi] \gamma \wedge [\pi'] \gamma \\
(A;) & [\pi;\pi'] \gamma \leftrightarrow [\pi][ \pi'] \gamma \\
(A^*) & [\pi^*] \gamma \leftrightarrow \gamma \wedge [\pi][\pi^*] \gamma\\
(A?) &  config(v \cdot  vs, mem) \to [v ?] config(vs,mem) \\

(A\neg ?) &  config(v \cdot  vs, mem) \to [v' ?] \gamma  \mbox{ where } v \mbox{ and } v' \mbox{ are distinct}.\\
\end{array}$

Next, we encode the transition system of the SMC machine as a set of axioms. Apart from the axioms for memory (which are straight-forward), we follow the rules of the SMC machine as closely as allowed by the formalism, using the same notation as in \cite{plotkin}. The sort of each variable can be easily deduced.

\medskip
\noindent{\bf SMC-inspired axioms}\nopagebreak 
 
$\begin{array}{ll}
(CStmt)& c(s1;s2)\lra c(s1);c(s2)\\ 
(AMem0) & empty \to get(x, 0)\\

(AMem1)  &  set(mem, x, n) \to  get(x, n)\\
(AMem2) & set(set(mem, x, n),y,m)\lra set(set(mem, y,m),x,n) \\
&\mbox{where } x \mbox{ and } y \mbox{ are distinct}\\
(AMem3) & set(set(mem, x, n),x,m)\lra set(mem, x,m)\\

(Aint) &   config(vs,mem) \to [c(n)] config(n \cdot vs,mem) \\
&\mbox{where } n \mbox{ is an integer}\\
(Aid)  &  config(vs, set(mem,x,n)) \to [c(x)] config(n \cdot vs,set(mem,x,n))\\
(Dplus) &    c(a1 + a2) \lra c(a1) ; c(a2) ; plus \\
(Aplus)  &  config(n2 \cdot n1 \cdot vs,mem) \to [plus] config(n \cdot vs,mem)\\
&\mbox{where } n \mbox{ is } n1 + n2 \\
(Dleq)   & c(a1 <= a2) \lra c(a2) ; c(a1) ; leq \\
(Aleq)    &config(n1 \cdot n2 \cdot vs,mem) \to  [leq] config(t \cdot vs,mem) \\
& \mbox{where } t \mbox{ is the truth value } \mbox{of } n1 \leq n2 \\
(Askip) &    \gamma \to  [c(skip)]\gamma \\
\end{array} $

$\begin{array}{ll}
(Dasgn) &   c(x := a) \lra c(a) ; asgn(x)\\
(Aasgn)  &  config(n \cdot vs,mem) \to [asgn(x)] config(vs,set(mem, x, n))\\
(Dif)  &  c(if\,\, b\,\, then\,\, s1 \,\,else\,\, s2) \lra c(b) ;
       ( (true \, ?; c(s1))\cup (false \, ?; c(s2)) ) \\
(Dwhile)  &  c(while\,\, b\,\, do\,\, s) \lra c(b) ; 
(true ? ; c(s); c(b))^*; false ? \\
\end{array}$

\medskip

Our  logical system is ${\mathbf K}\Lambda$, where $\Lambda$ contains all the PDL-inspired and the SMC-inspired axioms.
The general theory from the previous chapters provides two completeness results:
the completeness with respect to the canonical  frame (Theorem \ref{corhelp}) and the completeness with respect to the class of many-sorted boolean algebras with operators that are models of $\Lambda$ (Theorem \ref{alg1}). The canonical model is a \textit{non-standard} model with respect to the dynamic logic fragment of our system, meaning that the interpretation of the $c^*$ operation might not be the reflexive and transitive closure of the interpretation of $c$, where $c$ is a formula of sort \texttt{CtrlStack}.

Note that our PDL-inspired set of axioms does not include the \textit{induction axiom}:
\vspace*{-0.2cm}
\[ \gamma \wedge [\pi^*](\gamma\lra [\pi]\gamma )\lra [\pi^*]\gamma\]

 {The system presented in this paper can be used to certify executions, but we cannot perform symbolic verification. Adding the induction axiom would not solve this problem, the system should also be extended with quantifiers. Future research will adress this issue.}  
 
 \medskip

We conclude by a simple example formalizing and stating a formula which can be proven by deduction in our logic. Let $pgm$ be the following  program
\begin{center}

\texttt{i1:= 1; i2:= 2; if i1<=i2 then m:= i1 else m:= i2}
\end{center}
\noindent Note that $pgm$ is a formula of sort \texttt{Stmt} in our logic, $m$ is a formula of sort \texttt{Var} and $1$ is a formula of sort \texttt{Nat}. For this formula we can state and prove the following property:

\medskip

\noindent (P$_{pgm}$) $\,\,\vdash  config(vs,mem)\to[c(pgm)]config(vs,mem')$ implies \\ $\vds{\texttt{Mem}}\, mem'\to get(m,1)$ {for any  $mem, mem'$ of sort \texttt{Mem} and  $vs$  of sort \texttt{ValStack}.
}
\medskip

Which, can be read in plain English as: after executing $pgm$ the value of the program variable $m$ (in memory) will be $1$, and the value stack will be the same as before the execution.


Next we simplify the program \textit{pgm} to basic commands:

\medskip

$c(pgm) \lra c(i1 := 1) ; c(i2:=2) ; c(if i1<=i2 then m:= i1 else m:= i2) \lra$
\medskip

$\lra c(1) ; asgn(i1) ; c(2) ; asgn(i2) ; c(i1 <= i2) ;  $

\hspace*{3cm}$((true? ; c(m:=i1))\cup (false?;c(m:=i2))$

\medskip

$\lra c(1) ; asgn(i1) ; c(2) ; asgn(i2) ; c(i1);c(i2) ; leq ;  $

\hspace*{3cm}$((true? ; c(i1); asgn(m))\cup (false?;c(i2);asgn(m))$

\medskip

In the sequel, we give the main steps of the proof, each proof being a theorem of our ${\mathbf K}\Lambda$ system. Note that $Tranz$ denotes the following easily derived deduction rule: 

$(Tranz)$ \textit{if $\vds{s}\phi\to\psi$ and $\vds{s}\psi\to\chi$ then  $\vds{s}\psi\to \chi$}.

\medskip

{\bf Proof of} (P$_{pgm}$)

\medskip

(1) $config(vs,mem)\to [c(1)]config(1\cdot vs,mem)$  

\hfill $Aint$

(2) $config(1\cdot vs,mem) \to [asgn(i1)]config(vs,set(mem, i1,1))$  

\hfill $Aasgn$

(3) $[c(1)]config(1\cdot vs,mem)\to [c(1)][asgn(i1)]config(vs,set(mem,i1,1))$
 
 \hfill $UG(2),K$

(4) $[c(1)]config(1\cdot vs,mem)\to [c(1); asgn(i1)]config(vs,set(mem, i1,1))$

\hfill $A;$

(5) {$config(vs,mem)\to [c(1); asgn(i1)]config( vs,set(mem,i1,1))$}

\hfill $(1),(4),Tranz$

(6) $config(vs,set(mem,i1,1)) \to [c(2)]config(2\cdot vs, set(mem,i1,1))$ 

\hfill $Aint$

(7) $config(2\cdot vs, set(mem, i1,1)) \to$

\hspace*{2cm}$ \to[asgn(i2)]
config(vs, set(set(mem, i1,1), i2,2))$

\hfill  $Aasgn$

(8) $[c(2)]config(2\cdot vs, set(mem, i1,1)) \to$

\hspace*{2cm}$ \to [c(2)][asgn(i2)]
config(vs, set(set(mem, i1,1), i2,2))$

\hfill  $UG(7),K$

(9) {$config(vs,mem)\to [c(1); asgn(i1);$ }

\hspace{4cm} {$c(2);asgn(i2)]
config(vs, set(set(mem, i1,1), i2,2))$}

 \hfill {$UG (7),K,Tranz$}

(10) $config(vs, set(set(mem, i1,1), i2,2))\to$

\hfill $\to [c(i2)]config(2\cdot vs, set(set(mem, i1,1), i2,2)))$ 

\hfill $Aid$

(11) $config(vs, set(set(mem, i1,1), i2,2))\to$

\hfill $\to [c(i2)]config(2\cdot vs, set(set(mem, i2,2), i1,1)))$ 

\hfill $AMem2, UG(10),K$

(12) $config(2\cdot vs, set(set(mem, i2,2), i1,1)\to$

\hfill $\to [c(i1)]config(1\cdot 2\cdot vs, set(set(mem, i2,2), i1,1)))$  

\hfill $Aid$

(13) $config(vs, set(set(mem, i1,1), i2,2))\to$

\hfill $\to [c(i2)][c(i1)]config(1\cdot 2\cdot vs, set(set(mem, i2,2), i1,1)))$

\hfill $UG(12),K,Tranz$

(14) {$config(vs,mem)\to [c(1); asgn(i1);$ }

\hspace{4cm}{$c(2);asgn(i2)]$}

\hspace{4cm}{$[c(i2)][c(i1)]config(1\cdot 2\cdot vs, set(set(mem, i2,2), i1,1)))$}

 \hfill {$UG (13),K,Tranz$}

(15) $config(1\cdot 2\cdot vs, set(set(mem, i2,2), i1,1)))\to$

\hfill $\to [leq]config(true\cdot vs, set(set(mem, i2,2), i1,1))$ $Aleq$

(16) {$config(vs,mem)\to [c(1); asgn(i1);c(2);asgn(i2);c(i2);c(i1);leq]$}

\hspace{4cm}{$config(true\cdot vs, set(set(mem, i2,2), i1,1)))$}

 \hfill {$UG (15),K,Tranz$}

 (17) $config(true\cdot vs,set(set(mem, i2,2), i1,1)))\to$

\hfill $\to  [true?]config(vs, set(set(mem, i2,2), i1,1))$  

\hfill $A?$

(18) $config(vs, set(set(mem, i2,2), i1,1))\to$

\hfill $ [c(i1)]config(1\cdot vs, set(set(mem, i2,2), i1,1))$
 
\hfill $Aid$

(19) $config(1\cdot vs, set(set(mem, i2,2), i1,1]))\to$

\hfill $[asgn(m)]config(vs, set(set(set(mem, i2,2), i1,1),m,1))$ 

\hfill $Aasgn$

(20) $config(true\cdot vs,set(set(mem, i2,2), i1,1)))\to$

\hspace*{1cm}$[true?; c(i1); asgn(m)]config(vs, set(set(set(mem, i2,2), i1,1),m,1))$ 

\hfill $UG(17), UG(18), UG(19), K, Tranz$

(21) $config(true\cdot vs,set(set(mem, i2,2), i1,1)))\to$

\hspace*{1cm}$[false?][c(i2); asgn(m)]config(vs, set(set(set(mem, i2,2), i1,1),m,1))$ 

\hfill $A\neg ?$

(21') $config(true\cdot vs,set(set(mem, i2,2), i1,1)))\to$

\hspace*{1cm}$[false ? ; c(i2); asgn(m)]config(vs, set(set(set(mem, i2,2), i1,1),m,1))$ 

\hfill $A;$

(22) $config(true\cdot vs,set(set(mem, i2,2), i1,1)))\to$

\hspace*{1cm}$ [true?; c(i1); asgn(m)]config(vs, set(set(set(mem, i2,2), i1,1),m,1))\wedge $

\hspace*{1cm}$ [false ? ; c(i2); asgn(m)]config(vs, set(set(set(mem, i2,2), i1,1),m,1))$

\hfill  propositional logic

(22') $config(true\cdot vs,set(set(mem, i2,2), i1,1)))\to$

\hfill $[(true?; c(i1); asgn(m))\cup (false ? ; c(i2); asgn(m))]$ \hfill\;

\hfill $config(vs, set(set(set(mem, i2,2), i1,1),m,1)) $

\hfill  $A\cup$

(23) {$config(vs,mem)\to [c(pgm)] config(vs, set(set(set(mem, i2,2), i1,1),m,1)) $}

\hfill $UG(22'),Tranz$

\vspace*{0.5cm}

\medskip

We proved that

$\vdash config(vs,mem)\to [c(pgm)]config(vs, set(set(set(mem, i2,2), i1,1),m,1))$

and we know $\vdash set(set(set(mem, i2,2), i1,1),m,1)\to get(m,1)$ so

$\vdash config(vs,mem)\to [c(pgm)]config(vs, memf)$ implies $\vdash memf\to get(m,1)$

which ends our proof.

\medskip
\noindent\textit{Related Work.}
It has been shown~\cite{contextModal} that \textit{Separation Logic}~\cite{separation} (as branch-logic over the theory of maps) can be faithfully represented in modal logic; therefore one can employ dynamic logic reasoning directly on top of separation logic.  Instead of following that approach, we here follow more closely the line of work of \textit{Matching Logic}~\cite{rosu}, using directly the operational semantics rather than the traditional Floyd/Hoare axiomatic semantics.

 \newpage

\bibliographystyle{fundam}


\end{document}